\documentclass[conference]{IEEEtran}
\ifCLASSINFOpdf \else \fi

\usepackage{cbc}

\usepackage{graphicx}
\usepackage{physics}
\usepackage[hidelinks]{hyperref}
\usepackage{bm}
\usepackage{amsthm}
\usepackage{dsfont}
\usepackage{algorithm}
\usepackage{algpseudocode}

\theoremstyle{definition}

\newtheorem{proposition}{Proposition}

\usepackage{array}
\ifCLASSOPTIONcompsoc
 \usepackage[caption=false,font=normalsize,labelfont=sf,textfont=sf]{subfig}
\else
 \usepackage[caption=false,font=footnotesize]{subfig}
\fi
\usepackage{url}
\usepackage[noadjust]{cite}

\usepackage{afterpage}

\begin{document}

\title{Initial Access Optimization for RIS-assisted Millimeter Wave Wireless Networks}

\author{{Charbel Bou Chaaya}\textsuperscript{1}, {Mohamad Assaad}\textsuperscript{2} and {Tijani Chahed}\textsuperscript{1}\\
\textsuperscript{1} {\normalsize\textit{Institut Polytechnique de Paris, Télécom SudParis, 19 Place Marguerite Perey, 91120 Palaiseau, France}}\\
\textsuperscript{2} {\normalsize\textit{Université Paris-Saclay, CNRS, CentraleSupélec, Laboratoire des signaux et systèmes, 91190 Gif-sur-Yvette, France}}}

% \thanks{C. Noyes is with Villa Park High School, Villa Park,
% CA, 92861 USA e-mail: cnoyes@usc.edu.}% <-this % stops a space
% \thanks{Manuscript received December 00, 0000}}

% \markboth{IEEE Transactions on XXXXX,~Vol.~0, No.~0, December~0}%
% {Shell \MakeLowercase{\textit{et al.}}: Bare Demo of IEEEtran.cls for Journals}

\maketitle

% As a general rule, do not put math, special symbols or s
\begin{abstract}
Reconfigurable Intelligent Surfaces (RIS) are considered a key enabler to achieve the vision of Smart Radio Environments, where the propagation environment can be programmed and controlled to enhance the efficiency of wireless systems. These surfaces correspond to planar sheets comprising a large number of small and low-cost reflecting elements whose parameters are adaptively selected with a programmable controller. Hence, by optimizing these coefficients, the information signals can be directed in a customized fashion. On the other hand, the initial access procedure used in 5G is beam sweeping, where the base station sequentially changes the active beam direction in order to scan all users in the cell. This conventional protocol results in an initial access latency. The aim of this paper is to minimize this delay by optimizing the activated beams in each timeslot, while leveraging the presence of the RIS in the network. The problem is formulated as a hard optimization problem. We propose an efficient solution based on jointly alternating optimization and Semi Definite Relaxation (SDR) techniques. Numerical results are provided to assess the superiority of our scheme as compared to conventional beam sweeping.  
\end{abstract}

%\begin{IEEEkeywords}
%Reconfigurable Intelligent Surfaces (RIS), Millimeter Wave (mmWave), Initial Access, Beam Sweeping, Alternating Optimization, Semidefinite Programming.
%\end{IEEEkeywords}

\IEEEpeerreviewmaketitle

\section{Introduction}

\par To respond to the exponential surge for data traffic, the capacity of next-generation wireless networks needs to increase swiftly. A promising technology for 5G and Beyond wireless networks is the millimeter wave communication (mmWave) over the 30–300 GHz due to its potential high data rates thanks to their bandwidth wealth \cite{rangan2014millimeter}. To compensate the severe interference experienced on the congested and fragmented sub-6 GHz spectrum, large antenna arrays are used to form highly directional beams at these high frequencies. Thus, a principal characteristic of mmWave signals is their directionality, as they largely depend on the angles of departure and arrival (AoD/AoA) between the communicating devices. This high directivity is also their weakness, because it comes at the price of less resilience to signal blockages, and an intermittent channel quality. Accordingly, a massive amount of expensive antennas is packed at the base station (BS) to produce highly directional beams and channel gains \cite{HAC2016}.

\par Alternatively, reconfigurable intelligent surfaces (RIS), or intelligent reflecting surfaces (IRS), have been recently introduced as a wireless paradigm that can exploit engineered scattering surfaces to transmit and receive information \cite{liaskos2018new}. The core technology behind this new prototype is the meta-surface which is a planar array made of a large number of small scattering elements whose reflecting coefficients depend on their properties. Using low-cost fabrication techniques and integrated circuits, these properties can be easily controlled and each of these elements can then reflect a phase-shifted version of the impinging electromagnetic wave. Therefore, by optimizing the induced phase shifts, the radio propagation environment can be tailored to direct the signals towards the users, boosting the coverage area and avoiding blockages.

\par Deploying RIS in wireless systems has shown very promising advantages in terms of higher uplink and downlink data rates, and favorable energy consumption \cite{wu2019intelligent}. A missing aspect in the literature studies on RIS is their role in the control signaling share of wireless communications, for instance, the initial access process. In fact, 5G uses the technique of beam sweeping in the initial access procedure \cite{andrews2014will}. Beam sweeping is a transmission method used at the BS to determine suitable communications directions. During the initial access, the BS transmits cell-specific signals to the users while switching the beam direction in a sequential manner in order to scan the whole cell. These signals are synchronized, periodical, and include necessary initial access information to be conveyed. Consequently, while the BS shifts the beam direction at each timeslot, users in regions that are aligned with the active beams at each slot can access the network. The BS keeps on switching the beam direction in order to cover the whole cell, which causes an initial access latency.

\par This initial delay is due to the fact that the BS does not know where the users are located, and must figure the appropriate transmission directions. Added to that, since the mmWave links are very sensitive to any misalignment or blockage, the initial access protocol is repeated more frequently. Also, the presence of the RIS in the network brings new degrees of freedom in the channels that should be used to further optimize this initial procedure. 

\par In the literature, many aspects of beam sweeping have been studied. For example, \cite{barati2014directional} and \cite{barati2015directional} analyze a simple directional cell search process where the BS transmits activation signals in random directions to scan the angular space. The performance of the initial access strategy is shown to be dependent on the target $\SNR$ (Signal-to-Noise ratio) range in \cite{wei2017exhaustive}. The availability of cell context information at the BS was shown to reduce the search delay of analog beamforming in \cite{abbas2016context}. A heuristic weight-based beam sweeping algorithm to improve the cell search in mmWave 5G networks is given in \cite{perera2020initial}. In \cite{ly2021initial}, an initial access scheme is developed with the aim to minimize the initial access delay by optimizing the activated beams and regions at each slot. In this paper, we extend the previous framework of minimizing the initial access delay by considering the presence of RIS, which makes the problem much more challenging as, in addition to optimizing the beams, the phase shifts must also be properly optimized to enhance the performance of the initial access scheme.  

Our work can be seen as an extension to \cite{ly2021initial}, where the initial access delay is minimized by optimizing the activated beams and regions at each slot. The novel factor in our case is presence of the RIS that should be leveraged to achieve better results.

\par The remainder of this paper is organized as follows. In section \ref{sec:2}, the system model is provided and the initial access problem is formulated. This problem of interest is then elaborated and our proposed solution is explained in section \ref{sec:3}. Numerically evaluated simulation results are given in section \ref{sec:4}, while section \ref{sec:5} concludes the paper.

\section{System Model and Problem Formulation} \label{sec:2}
\subsection{System Model}
\par We consider the downlink of a mmWave wireless system, where the BS with $N_a$ antennas is seeking to activate all the single antenna users in the cell, by choosing the best beamforming directions, as shown in \autoref{fig:system}. Accordingly, the cell is partitioned into $Q$ non-overlapping regions, that are small enough to efficiently cover the whole cell. For each zone $q \in \{1,\dots,Q\}$, we denote by $w_q$ the expected number of users in this area. We assume that users in each region share the same second order channel statistics (SOCS); meaning that they are simultaneously covered or not covered by the active set of beams. We then define an indicator function for each $q \in \{1,\dots,Q\}$, $r_q = \mathds{1}_{\{\text{zone} \, q \, \text{is covered}\}}$. At each timeslot, the BS selects a group of beams from a fixed set of $N$ beams $\{\bm{b}_1, \dots, \bm{b}_N\}$, where $\bm{b}_n$ is the $n^{\text{th}}$ beam. We define $y_n$ as an indicator function for these beams as: $y_n = \mathds{1}_{\{\bm{b}_n \text{is chosen}\}}$. The RIS has $M$ reflecting elements and its position is fixed at the edge of the cell. Note that the expected number of users in all locations of the cell can be computed based on the users' distribution, usually known (e.g. \cite{li2016heterogeneous}).

\par We say that a certain zone is activated by the BS if the average received $\SNR$ is higher than a predetermined threshold $\tau\geq0$. For each region $q \in \{1,\dots,Q\}$, this is expressed as:
\begin{equation}
    \SNR_q = \frac{\sum_{n=1}^N \abs{\left(\bm{h}^\herm_{r,q} \, \bm{\Theta} \, \bm{G} + \bm{h}^\herm_{d,q}\right)
    \bm{b}_n}^2 \, y_n}{\sigma^2 \sum_{n=1}^N y_n}
    \geq \tau \, r_q
\end{equation}
where $\bm{h}^\herm_{d,q} \in \C^{1 \times N_a}, \bm{G} \in \C^{M \times N_a}$ and $\bm{h}^\herm_{r,q} \in \C^{1 \times M}$ represent the channels between the BS and zone $q$, the BS and the RIS, the RIS and zone $q$ respectively. The RIS reflection matrix $\bm{\Theta}$ is a diagonal matrix given by $\bm{\Theta} = \diag\left(\left[e^{j\theta_1}, \dots, e^{j\theta_M}\right]\right)$, where $\bm{\theta} = [\theta_1, \theta_2, \dots, \theta_M]^\trans \in [0, 2\pi]^M$ are the phase shifts of the $M$ elements. Also, $\sigma^2$ is the power of thermal noise, and $\sum_{n=1}^N y_n$ is a normalization factor since the BS can activate multiples beams at the same time.

\begin{figure}
    \centering
    \includegraphics[width=\columnwidth]{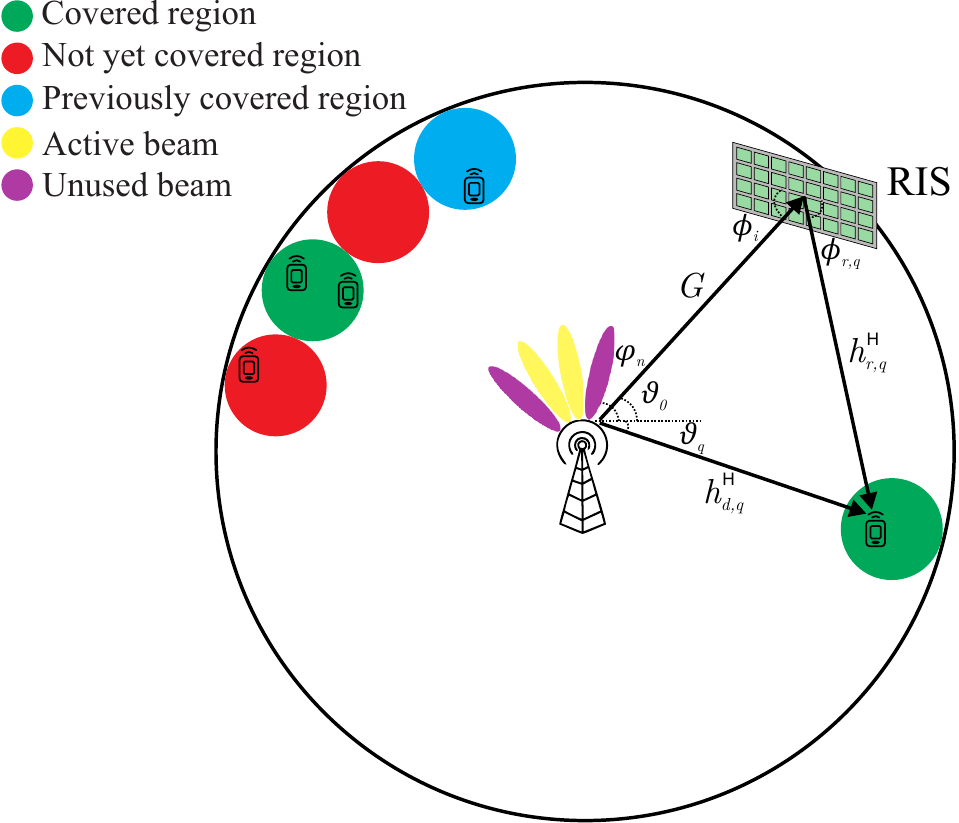}
    \caption{System Model}
    \label{fig:system}
\end{figure}

\par By carefully positioning the RIS, the mmWave propagation environment can be fully characterized by line-of-sight (LoS) channels. These LoS channels are modeled as the product of the array responses of the two nodes of the channel. For convenience, we define the steering vector function:
\begin{equation}
    \bm{a}(\phi, N) = \frac{1}{\sqrt{N}} \left[ 1, e^{j \pi \phi}, \dots, e^{j \pi (N-1) \phi}\right]^\trans
\end{equation}
where $N$ denotes the number of elements in the array and $\phi$ is the phase difference between the observations at two adjacent elements. One can easily see that $\bm{a}$ is a $2$-periodic function of $\phi$, and thus, we restrict $\phi$ to $[0,2)$ by replacing it with $\phi - 2\floor{\frac{\phi}{2}}$ when needed. We also adopt the distance dependent path loss model given by:
\begin{equation}
    L(d) = C_0 \left(\frac{d}{D_0}\right)^{-\alpha}
\end{equation}
where $C_0$ is the reference path loss at distance $D_0$, $\alpha$ is the decay exponent, and $d$ is the distance between the two nodes. We can now write the beams at the BS and the channels as:
\begin{align}
    \bm{b}_n &=&& \sqrt{P_t} \, \bm{a}\left(2\frac{d_a}{\lambda} \cos(\varphi_n), N_a\right) \\
    \bm{h}^\herm_{d,q} &=&& \sqrt{N_a \, L(d_q)} \; \bm{a}^\herm\left(2\frac{d_a}{\lambda} \cos(\vartheta_q), N_a\right) \\
    \bm{h}^\herm_{r,q} &=&& \sqrt{M \, L(d_{r,q})} \; \bm{a}^\herm\left(2\frac{d_m}{\lambda} \cos(\phi_{r,q}), M\right) \\
    \bm{G} &=&& \sqrt{M \, N_a \, L(d_0)} \nonumber \\ &&&\bm{a}\left(2\frac{d_m}{\lambda} \cos(\phi_i), M\right) \; \bm{a}^\herm\left(2\frac{d_a}{\lambda} \cos(\vartheta_0), N_a\right)
\end{align}
where $P_t$ is the BS transmit power, $\lambda$ is the carrier's wavelength, $d_a$ and $d_m$ are the respective antenna spacing at the BS and the RIS fixed at half-wavelength, $d_0$, $d_q$ and $d_{r,q}$ are the distances between the BS and the RIS, the BS and zone $q$, and the RIS and zone $q$ respectively. $\varphi_n$ is the physical AoD of the $n^\text{th}$ beam at the BS. $\vartheta_0, \vartheta_q$ and $\phi_{r,q}$ are the physical AoD from the BS to the RIS, from the BS to zone $q$, and from the RIS to zone $q$ respectively. $\phi_i$ is the physical AoA at the RIS from the BS. \autoref{fig:system} also displays these various angles.

\subsection{Problem Formulation}
\par While noting that the BS manages the RIS controller to set the phase shifts at each slot through a control link, the problem of initial access boils down to minimizing the number of timeslots needed to activate all the $Q$ regions. In other words, at each timeslot, the BS must find a group of beams and RIS phase shifts in order to raise the $\SNR$ of a maximum number of zones above the threshold. In the next slot, the BS repeats the process while disregarding previously covered areas. And hence, the optimization problem to solve at a given timeslot, is the following:
\begin{equation*}
    \label{P1}
    \begin{aligned}
        &\textrm{(P1)}& \quad & \underset{\{y_n\}, \{r_q\}, \bm{\theta}}{\textrm{maximize}} \quad && \sum_{q=1}^Q r_q w_q\\
        &&& \textrm{subject to} \quad && \SNR_q \geq \tau \, r_q &&\quad q=1, \dots, Q \\
        &&&&& 0 \leq \theta_i \leq 2 \pi &&\quad i=1, \dots, M  \\
        &&&&& r_q, \, y_n \in \{0, 1\} &&\quad q=1, \dots, Q, \\ &&&&&&&\quad n=1, \dots, N.
    \end{aligned}
\end{equation*}

\section{Proposed Solution} \label{sec:3}
\par The aforementioned problem is a mixed integer nonlinear program (MINLP), and thus falls in the class of NP-hard problems that cannot be optimally solved in polynomial time. In fact, the challenging factor is the joint optimization over the activated beams and regions, and the choice of suitable RIS phases. Therefore, our proposed solution is to alternately solve two problems: the first is to find the maximum number of regions that can be covered by a set of active beams, given the RIS phase shifts; and the second is to optimize the phase shifts to activate more regions.

\subsection{First sub-problem}
\par Indeed, the first sub-problem is the typical initial access problem for mmWave wireless networks without a RIS. This problem is similar to the one studied in \cite{ly2021initial} (Section III Problem P2). The only aspect that changes in our case is the more complex channel induced by the RIS phase shifts, which are fixed during the solution of this first problem.

\par We will directly show the solution for this problem. All derivation details are omitted for brevity, as they are similar to the ones in \cite{ly2021initial}. We start by defining $\bm{v} = [y_1, \dots, y_N, r_1, \dots, r_Q, 1]^\trans$ and $S=N+Q$. Let $\bm{V} = \bm{vv}^\trans$. The first sub-problem problem was shown to be equivalent to:
\begin{equation*}
    \label{P2}
    \begin{aligned}
        &\textrm{(P2)}& \quad & \underset{\bm{V}}{\textrm{maximize}} \quad && \tr(\bm{FV})\\
        &&& \textrm{subject to} \quad && \tr((\bm{A}_q-\bm{C}_q)\bm{V}) \geq 0 \quad q=1, \dots, Q\\
        &&&&& \tr(\bm{H}_p \bm{V}) = 0 \qquad\qquad p=1, \dots, S \\
        &&&&& \left[\bm{V}\right]_{S+1,S+1}=1, \\
        &&&&& \bm{V} \succeq \bm{0}, \; \rank\left(\bm{V}\right) = 1
    \end{aligned}
\end{equation*}
where
\begin{flalign}
    &\bm{A}_q = 
    \begin{bmatrix}
        \bm{0}_{S \times S} & & \bm{b}_q \\
        \bm{b}_q^\trans & \bm{0}_{1 \times Q} & \bm{0}_{Q+1 \times 1}
    \end{bmatrix},
    \\
    &\bm{b}_q = \frac{1}{2} \left[\abs{\bm{h}_q^\herm \, \bm{b}_1}^2, \dots, \abs{\bm{h}_q^\herm \, \bm{b}_N}^2\right]^\trans,
    \\
    &\bm{h}_q^\herm = \bm{h}^\herm_{r,q} \, \bm{\Theta} \, \bm{G} + \bm{h}^\herm_{d,q},
    \end{flalign}
\begin{align}
    &\begin{multlined}
        \left[\bm{C}_q\right]_{i,j} = \frac{\tau \sigma^2}{2} \Bigl(\mathds{1}_{\{i=N+q\}}\times\mathds{1}_{\{1\leq j \leq N\}} +\\
        \mathds{1}_{\{1 \leq i \leq N\}}\times\mathds{1}_{\{j=N+q\}}\Bigr),
    \end{multlined}
\\
    &\begin{multlined}
        \left[\bm{H}_p\right]_{i,j} = \Bigl(\mathds{1}_{\{i=p\}}\times\mathds{1}_{\{j=p\}}-\\ \frac{1}{2}\mathds{1}_{\{i=p\}}\times\mathds{1}_{\{j=S+1\}}-\frac{1}{2}\mathds{1}_{\{i=S+1\}}\times\mathds{1}_{\{j=p\}}\Bigr),
    \end{multlined}
    \\
    &\bm{F} = \diag\left(\left[\bm{0}_{1 \times N}, w_1, \dots, w_Q, 0\right]^\trans\right).
\end{align}

\par The obtained problem is again generally NP-hard, because of the non-convex rank constraint, whereas all other constraints are convex. The proposed solution is to use Semidefinite Relaxation (SDR) as follows: we drop the rank constraint and obtain a standard semidefinite program (SDP) that can be optimally solved, then we build a suboptimal rank-one solution to the original problem. If the optimal solution $\bm{V}^\star$ for the problem without rank constraint is rank-one, then it is also the optimal solution to \hyperref[P2]{(P2)}. Otherwise, we use Gaussian randomization \cite{ye1998semidefinite}, \cite{luo2007approximation} to approximate the optimal solution by generating random feasible solutions from $\bm{V}^\star$, and then recover their binariness using a threshold. \hyperref[alg: P2]{Algorithm 1} summarizes this method.

\begin{algorithm}
    \caption{SDR based solving approach for \hyperref[P2]{(P2)}}
    \label{alg: P2}
    \begin{algorithmic}
        \Require $\bm{F}$, phase shifts: $\bm{\theta}$, number of randomizations: $L_1$ 
        \State{Solve \hyperref[P2]{(P2)} without rank constraint optimally and denote $\bm{V}^\star$ the solution}
        \State{Extract the $S \times S$ left upper sub-matrix $\bm{V}^\prime$ from $\bm{V}^\star$}
        \For{i = $1$ \text{to} $L_1$}
            \State{Generate $\hat{\bm{v}}^{(i)} \sim \mathcal{N}\left(\bm{0}_{S \times 1},\bm{V}^\prime\right)$}
            \For{j = $1$ \text{to} $S$}
                \If{$\hat{\bm{v}}^{(i)}_j < 0$}
                    \State $\hat{\bm{v}}^{(i)}_j \gets 0$
                \Else \State $\hat{\bm{v}}^{(i)}_j \gets 1$
                \EndIf
            \EndFor
        \EndFor
        \State Among $L_1$ generated vectors $\hat{\bm{v}}$, find $\bm{v}$ that is feasible and maximizes $\tr(\bm{FV})$\\
        \hspace{-1em}\textbf{Output: $\bm{v}$}
    \end{algorithmic}
\end{algorithm}

\subsection{Second sub-problem}
\par The second sub-problem is the maximization over $\bm{\theta}$ of the weighted sum of the activated regions obtained from the first sub-problem, under $\SNR$ constraints. After obtaining $\bm{v}$, due to the independence between the phase shifts and the objective function, this problem is reduced to a feasibility check. We start by reformulating it as follows.

\par We define $\bm{t} = \left[e^{j\theta_1}, \dots, e^{j\theta_M}\right]^\herm$ and note $\bm{\Omega}_{q} = \diag(\bm{h}_{r,q}^\herm) \, \bm{G}$. We then have $\bm{h}^\herm_{r,q} \, \bm{\Theta} \, \bm{G} = \bm{t}^\herm \, \bm{\Omega}_{q}$, and the problem becomes:
\begin{equation*}
    \begin{aligned}
        &\textrm{Find} \quad && \bm{t}\\
        &\textrm{subject to} \quad && \sum_{n=1}^N y_n \abs{\left(\bm{t}^\herm \bm{\Omega}_q + \bm{h}^\herm_{d,q}\right)\bm{b}_n}^2 \geq \tau \sigma^2 r_q \sum_{n=1}^N y_n \quad \forall q\\
        &&& \abs{\bm{t}_i} = 1 \qquad\qquad i=1, \dots, M+1.
    \end{aligned}
\end{equation*}
We notice that the two constraints can be transformed to quadratic constraints as follows:
\begin{align}
\begin{split}
    &\abs{\left(\bm{h}^\herm_{r,q} \, \bm{\Theta} \, \bm{G} + \bm{h}^\herm_{d,q}\right)\bm{b}_n}^2
    = \abs{\left(\bm{t}^\herm \bm{\Omega}_q + \bm{h}^\herm_{d,q}\right)\bm{b}_n}^2 \\
    &\begin{multlined}
        = \bm{t}^\herm \, \bm{\Omega}_q \, \bm{b}_n \, \bm{b}_n^\herm \, \bm{\Omega}_q^\herm \, \bm{t}
        + \bm{t}^\herm \, \bm{\Omega}_q \, \bm{b}_n \, \bm{b}_n^\herm \, \bm{h}_{d,q} \\
        + \bm{b}_n^\herm \, \bm{\Omega}_q^\herm \, \bm{t} \, \bm{h}_{d,q}^\herm \, \bm{b}_n
        + \bm{b}_n^\herm \, \bm{h}_{d,q} \, \bm{h}_{d,q}^\herm \, \bm{b}_n
    \end{multlined}\\
    &= \overline{\bm{t}}^\herm \, \bm{Z}_{q,n} \, \overline{\bm{t}} = \tr\left(\bm{Z}_{q,n} \, \bm{T}\right)
\end{split}
\end{align}
where
\begin{equation}
    \bm{Z}_{q,n} = 
    \begin{bmatrix}
        \bm{\Omega}_q \, \bm{b}_n \, \bm{b}_n^\herm \, \bm{\Omega}_q^\herm
        & \bm{b}_n^\herm \, \bm{h}_{d,q} \, \bm{\Omega}_q \, \bm{b}_n \\
        \bm{h}_{d,q}^\herm \, \bm{b}_n \, \bm{b}_n^\herm \, \bm{\Omega}_q^\herm
        & \abs{\bm{h}_{d,q}^\herm \, \bm{b}_n}^2
    \end{bmatrix}
    \quad
    \overline{\bm{t}} = 
    \begin{bmatrix}
        \bm{t} \\ 1
    \end{bmatrix}
\end{equation}
and $\bm{T} = \overline{\bm{t}} \, \overline{\bm{t}}^\herm$. Now the problem can be rewritten as:
\begin{equation*}
    \label{P3}
    \begin{aligned}
        &\textrm{(P3)}& \quad &\textrm{Find} \quad && \bm{T}\\
        &&&\textrm{subject to} \quad && \sum_{n=1}^N y_n \tr\left(\bm{Z}_{q,n} \, \bm{T}\right) \geq \tau \sigma^2 r_q \sum_{n=1}^N y_n  \; \forall q\\
        &&&&& \left[\bm{T}\right]_{i,i} = 1 \qquad\qquad i=1, \dots, M+1 \\
        &&&&& \bm{T} \succeq \bm{0}, \; \rank\left(\bm{T}\right) = 1.
    \end{aligned}
\end{equation*}

\hyperref[P3]{(P3)} contains a hard non-convex rank constraint. Using SDR to relax this constraint, we obtain an SDP that can be optimally solved, and a feasible solution of \hyperref[P3]{(P3)} can be extracted from the higher rank solution as shown in \hyperref[alg: P3]{Algorithm 2} \cite{wu2018intelligent}.

\begin{algorithm}
    \caption{SDR based solving approach for \hyperref[P3]{(P3)}~/~\hyperref[P4]{(P4)}}
    \label{alg: P3}
    \begin{algorithmic}
        \Require $\bm{v}$, number of randomizations: $L_2$
        \State{Solve \hyperref[P3]{(P3)}~/~\hyperref[P4]{(P4)} without rank constraint optimally and denote $\bm{T}^\star$ the solution}
        \State{Compute the EVD: $\bm{T} = \bm{U \, \Sigma \, U}^\herm$ where $\bm{U}$ and $\bm{\Sigma}$ are a unitary and a diagonal matrix of size $M+1$}
        \For{i = $1$ \text{to} $L_2$}
            \State{Generate $\bm{r} \sim \mathcal{CN}\left(\bm{0}_{(M+1) \times 1}, \bm{I}_{M+1}\right)$}
            \State{$\hat{\bm{t}}^{(i)} \gets \bm{U \, \Sigma}^{1/2} \, \bm{r}$}
        \EndFor
        \State Among $L_2$ generated vectors $\hat{\bm{t}}$, find $\tilde{\bm{t}}$ that is feasible if \hyperref[P3]{(P3)} is considered (or that is feasible and maximizes the objective function in the case of \hyperref[P4]{(P4)})
        \For{j = $1$ \text{to} $M$}
            \State{$\bm{t}_j \gets \exp\left(j \arg \left( \frac{\tilde{\bm{t}}_j}{\tilde{\bm{t}}_{M+1}} \right)\right)$}
        \EndFor \\
        \hspace{-1em}\textbf{Output: $\bm{t}$}
    \end{algorithmic}
\end{algorithm}

Although solving \hyperref[P3]{(P3)} yields RIS phase shifts that satisfy the $\SNR$ constraint for the activated regions by the BS, this solution does not actually use the RIS to activate more regions than those covered without an RIS. And since this is the advantage of introducing metasurfaces, we go a step further and propose a new problem to solve in the second step to leverage the RIS for improving the activated zones in a timeslot. To do so, we introduce $Q$ auxiliary variables $\{\gamma_q\}$. Those variables represent the residual $\SNR$ in each inactive region. We propose now to maximize the weighted sum of these residual $\SNR$ variables as follows:
\begin{equation*}\begin{aligned}
    \label{P4}
        &\textrm{(P4)}\qquad\underset{\bm{T}, \{\gamma_q\}}{\textrm{maximize}}\qquad \sum_{q=1}^Q (1 - r_q) \, w_q \, \gamma_q\\
        &\textrm{s. t.}\quad \sum_{n=1}^N y_n \tr\left(\bm{Z}_{q,n} \, \bm{T}\right) \geq \tau \sigma^2 r_q \sum_{n=1}^N y_n + (1 - r_q) \, \gamma_q \quad\forall q\\
        &\qquad\;\;\left[\bm{T}\right]_{i,i} = 1 \qquad\qquad i=1, \dots, M+1 \\
        &\qquad\;\;\bm{T} \succeq \bm{0}, \; \rank\left(\bm{T}\right) = 1.
\end{aligned}\end{equation*}
The  idea behind this formulation is to find phase shifts for the RIS, that maximize the $\SNR$ for the inactive regions, while keeping the $\SNR$ level in the activated regions from the first step above the threshold. Also, both problems have the same feasible set of $\bm{T}$. A similar SDR-based solution can be obtained for this problem with a variation shown in \hyperref[alg: P3]{Algorithm 2}.

\subsection{Alternating Optimization}
\par Now we can finally present \hyperref[alg: alt]{Algorithm 3} to solve \hyperref[P1]{(P1)}.
\begin{algorithm}
    \caption{Alternating Optimization}
    \label{alg: alt}
    \begin{algorithmic}
        \Require Initial phase shifts $\bm{\theta}^1$
        \State Set iteration index $r \gets 1$ 
        \Repeat{
        \\ \hspace{1em} Solve \hyperref[P2]{(P2)} for a given $\bm{\theta}^r$, without rank constraint, 
        \\ \hspace{1.5em} using \hyperref[alg: P2]{Algorithm 1} to get $\bm{v}^r$
        \\ \hspace{1em} Solve \hyperref[P3]{(P3)}~/~\hyperref[P4]{(P4)} for a given $\bm{v}^r$, without rank constraint,
        \\ \hspace{1.5em} using \hyperref[alg: P3]{Algorithm 2} to get $\bm{\theta}^{r+1}$
        \\ \hspace{1em} Update $r \gets r+1$}
        \Until{Convergence or a stopping criterion is reached}
    \end{algorithmic}
\end{algorithm}

\par After presenting the proposed two-stage method, the last step would be to prove that this algorithm converges to a solution. This is ensured by the following proposition.

\begin{proposition}
\hyperref[alg: alt]{Algorithm 3} monotonically increases the objective value of \hyperref[P1]{(P1)} and converges.
\end{proposition}

\begin{proof}
The objective function $f(\bm{v}, \bm{\theta}) = \sum_{q=1}^Q r_q w_q$ is non-decreasing with the iteration of this algorithm, since at any consecutive iterations $r$ and $r+1$, we have:
\begin{equation}
    f(\bm{v}^{r+1}, \bm{\theta}^{r+1}) \overset{\text{(a)}}{\geq} f(\bm{v}^{r}, \bm{\theta}^{r+1}) \overset{\text{(b)}}{=} f(\bm{v}^{r}, \bm{\theta}^{r})
\end{equation}
where (a) results from $\bm{v}^{r+1}$ being the optimal solution given $\bm{\theta}^{r+1}$ and (b) holds since $f$ is independent of $\bm{\theta}$.
\end{proof}

Finally, the proposed initial access algorithm is given in the following \hyperref[alg:acc]{Algorithm 4}.
\begin{algorithm}
    \caption{Proposed initial access with RIS}
    \label{alg:acc}
    \begin{algorithmic}
        \Require $\bm{F}$
        \State Set timeslot $\gets 0$ 
        \While{$\tr\left(\bm{F}\right) > 0$}
        \State Solve \hyperref[P1]{(P1)} using \hyperref[alg: alt]{Algorithm 3}
        \For{q = 1 \text{to} Q}
            \If{$q^{\text{th}}$ zone activated}
            $[\bm{F}]_{q,q} \gets 0$
            \EndIf
        \EndFor
        \State Increase timeslot $\gets$ timeslot $+1$
        \EndWhile \\
        \hspace{-1em}\textbf{Output: }timeslot
    \end{algorithmic}
\end{algorithm}

\section{Numerical Results} \label{sec:4}
\par In this section, we report some computer simulation results to illustrate the efficiency of the proposed scheme. We consider a circular cell populated by $100$ users and divided into $Q = 40$ regions. The viability of mmWave cellular up to cell radii on the order of hundreds of meters is shown in \cite{andrews2016modeling}, and peak data rate was obtained over a range up to almost $2$ km in LoS scenario. We thus confidently pick the cell radius as $1$ km. The basestation is equipped with $N=N_a=64$ antennas and the RIS has $M=64$ elements. The transmit power is fixed at $P_t=1$ W, the reference path loss $C_0=-30$ dB, the decay exponent $\alpha=2.75$ and the noise power $\sigma^2=-85$ dBm.

\begin{figure}[h]
    \centering
    \includegraphics[width=8.5cm, height=6.5cm]{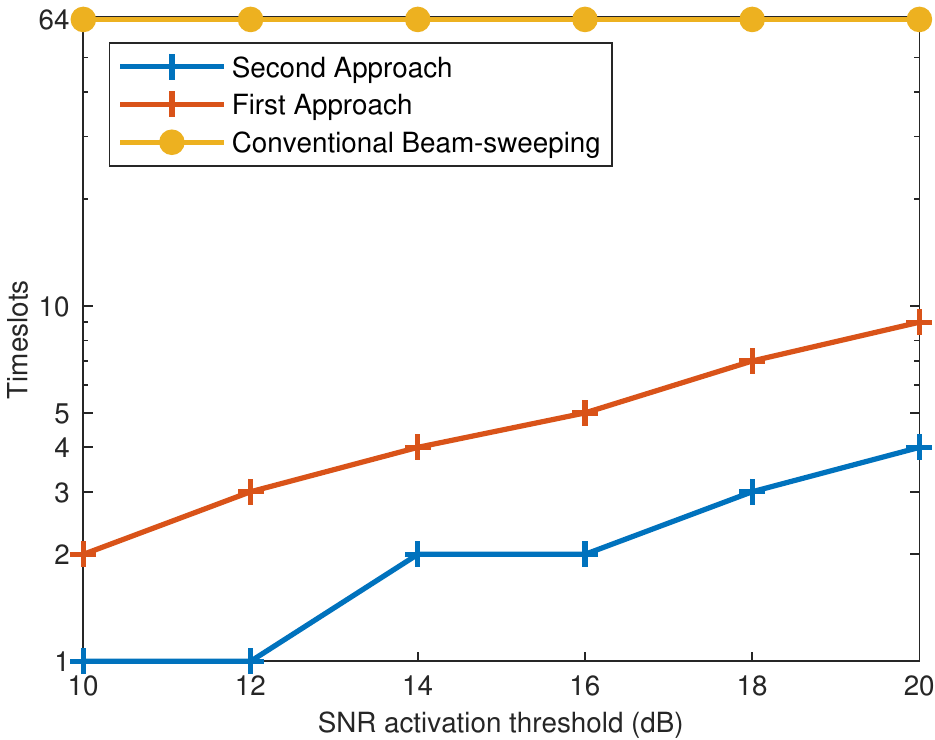}
    \caption{Initial access simulation results.}
    \label{fig:sim_ini}
\end{figure}

\autoref{fig:sim_ini} shows the number of timeslots needed to activate all users as a function of the $\SNR$ threshold that varies from $10$ to $20$ dB. The first and second approaches correspond to solving \hyperref[P3]{(P3)} and \hyperref[P4]{(P4)} in the second step, as discussed in the previous section. Clearly the second one provides less latency since it activates more zones in each slot compared to the first one. Also note that both outperform the conventional beam sweeping technique (in absence of RIS) that requires a number of timeslots equal to the number of BS beams (64 in this case), typically very large. 

\section{Conclusion} \label{sec:5}
\par In this paper, we provided an initial access scheme for RIS-assisted wireless systems with the goal of minimizing the required number of timeslots for the BS to cover all users in the cell. By using alternating optimization, the problem of interest is formulated and cut into two sub-problems. The solution to the first sub-problem consists in finding the beams for given RIS phase shifts, while an SDR-based solution is provided for the second sub-problem to optimize the RIS phase shifts. Simulation results show a significant improvement in the number of slots necessary for initial access compared to the traditional beam sweeping.

%\newpage

\bibliographystyle{IEEEtran}
%\bibliography{References.bib}

\end{document}